\documentclass[copyright,creativecommons]{eptcs}
\usepackage{breakurl}             % Not needed if you use pdflatex only.

\newcommand{\thenetherlands}{%%% 
The Netherlands}

\newcommand{\ourinstitute}{%%% 
Radboud University Nijmegen}

\newcommand{\ouraffiliation}{%%% 
Institute for Computing and Information Sciences \\
Heyendaalseweg 135, 6525 AJ Nijmegen\\
\thenetherlands}

\newcommand{\tueinstitute}{%%% 
Technical University Eindhoven\\
\thenetherlands}

\newcommand{\ourtitle}{%%% 
Pure Type Systems without Explicit Contexts}

\title{\ourtitle}

\def\titlerunning{\ourtitle}
\def\authorrunning{H. Geuvers, R. Krebbers, J. McKinna \&\ F. Wiedijk}

\author{%%%
  Herman Geuvers
  \institute{\ourinstitute} 
  \institute{\tueinstitute} \and 
  Robbert Krebbers\qquad\qquad 
  James McKinna\qquad\qquad 
  Freek Wiedijk\institute{\ourinstitute\\\ouraffiliation} 
} 

   \usepackage{amsthm}
   \newtheorem{definition}{Definition} 
   \newtheorem{proposition}[definition]{Proposition} 
   \newtheorem{lemma}[definition]{Lemma} 
   \newtheorem{theorem}[definition]{Theorem} 
   \newtheorem{corollary}[definition]{Corollary} 

\usepackage{amssymb,url,alltt,graphicx}
\def\bar{\mathrel{|}}

\newdimen\proofrulebreadth \proofrulebreadth=.05em
\newdimen\proofdotseparation \proofdotseparation=1.25ex
\newdimen\proofrulebaseline \proofrulebaseline=2ex
\newcount\proofdotnumber \proofdotnumber=3
\let\then\relax
\def\hfi{\hskip0pt plus.0001fil}
\mathchardef\squigto="3A3B
\newif\ifinsideprooftree\insideprooftreefalse
\newif\ifonleftofproofrule\onleftofproofrulefalse
\newif\ifproofdots\proofdotsfalse
\newif\ifdoubleproof\doubleprooffalse
\let\wereinproofbit\relax
\newdimen\shortenproofleft
\newdimen\shortenproofright
\newdimen\proofbelowshift
\newbox\proofabove
\newbox\proofbelow
\newbox\proofrulename
\def\shiftproofbelow{\let\next\relax\afterassignment\setshiftproofbelow\dimen0 }
\def\shiftproofbelowneg{\def\next{\multiply\dimen0 by-1 }%
\afterassignment\setshiftproofbelow\dimen0 }
\def\setshiftproofbelow{\next\proofbelowshift=\dimen0 }
\def\setproofrulebreadth{\proofrulebreadth}

\def\prooftree{% NESTED ZERO (\ifonleftofproofrule)
\ifnum  \lastpenalty=1
\then   \unpenalty
\else   \onleftofproofrulefalse
\fi
\ifonleftofproofrule
\else   \ifinsideprooftree
        \then   \hskip.5em plus1fil
        \fi
\fi
\bgroup% NESTED ONE (\proofbelow, \proofrulename, \proofabove,
\setbox\proofbelow=\hbox{}\setbox\proofrulename=\hbox{}%
\let\justifies\proofover\let\leadsto\proofoverdots\let\Justifies\proofoverdbl
\let\using\proofusing\let\[\prooftree
\ifinsideprooftree\let\]\endprooftree\fi
\proofdotsfalse\doubleprooffalse
\let\thickness\setproofrulebreadth
\let\shiftright\shiftproofbelow \let\shift\shiftproofbelow
\let\shiftleft\shiftproofbelowneg
\let\ifwasinsideprooftree\ifinsideprooftree
\insideprooftreetrue
\setbox\proofabove=\hbox\bgroup$\displaystyle % NESTED TWO
\let\wereinproofbit\prooftree
\shortenproofleft=0pt \shortenproofright=0pt \proofbelowshift=0pt
\onleftofproofruletrue\penalty1
}

\def\eproofbit{% NESTED TWO
\ifx    \wereinproofbit\prooftree
\then   \ifcase \lastpenalty
        \then   \shortenproofright=0pt  % 0: some other object, no indentation
        \or     \unpenalty\hfil         % 1: empty hypotheses, just glue
        \or     \unpenalty\unskip       % 2: just had a tree, remove glue
        \else   \shortenproofright=0pt  % eh?
        \fi
\fi
\global\dimen0=\shortenproofleft
\global\dimen1=\shortenproofright
\global\dimen2=\proofrulebreadth
\global\dimen3=\proofbelowshift
\global\dimen4=\proofdotseparation
\global\count255=\proofdotnumber
$\egroup  % NESTED ONE
\shortenproofleft=\dimen0
\shortenproofright=\dimen1
\proofrulebreadth=\dimen2
\proofbelowshift=\dimen3
\proofdotseparation=\dimen4
\proofdotnumber=\count255
}

\def\proofover{% NESTED TWO
\eproofbit % NESTED ONE
\setbox\proofbelow=\hbox\bgroup % NESTED TWO
\let\wereinproofbit\proofover
$\displaystyle
}%
\def\proofoverdbl{% NESTED TWO
\eproofbit % NESTED ONE
\doubleprooftrue
\setbox\proofbelow=\hbox\bgroup % NESTED TWO
\let\wereinproofbit\proofoverdbl
$\displaystyle
}%
\def\proofoverdots{% NESTED TWO
\eproofbit % NESTED ONE
\proofdotstrue
\setbox\proofbelow=\hbox\bgroup % NESTED TWO
\let\wereinproofbit\proofoverdots
$\displaystyle
}%
\def\proofusing{% NESTED TWO
\eproofbit % NESTED ONE
\setbox\proofrulename=\hbox\bgroup % NESTED TWO
\let\wereinproofbit\proofusing
\kern0.3em$
}

\def\endprooftree{% NESTED TWO
\eproofbit % NESTED ONE
  \dimen5 =0pt% spread of hypotheses
\dimen0=\wd\proofabove \advance\dimen0-\shortenproofleft
\advance\dimen0-\shortenproofright
\dimen1=.5\dimen0 \advance\dimen1-.5\wd\proofbelow
\dimen4=\dimen1
\advance\dimen1\proofbelowshift \advance\dimen4-\proofbelowshift
\ifdim  \dimen1<0pt
\then   \advance\shortenproofleft\dimen1
        \advance\dimen0-\dimen1
        \dimen1=0pt
        \ifdim  \shortenproofleft<0pt
        \then   \setbox\proofabove=\hbox{%
                        \kern-\shortenproofleft\unhbox\proofabove}%
                \shortenproofleft=0pt
        \fi
\fi
\ifdim  \dimen4<0pt
\then   \advance\shortenproofright\dimen4
        \advance\dimen0-\dimen4
        \dimen4=0pt
\fi
\ifdim  \shortenproofright<\wd\proofrulename
\then   \shortenproofright=\wd\proofrulename
\fi
\dimen2=\shortenproofleft \advance\dimen2 by\dimen1
\dimen3=\shortenproofright\advance\dimen3 by\dimen4
\ifproofdots
\then
        \dimen6=\shortenproofleft \advance\dimen6 .5\dimen0
        \setbox1=\vbox to\proofdotseparation{\vss\hbox{$\cdot$}\vss}%
        \setbox0=\hbox{%
                \advance\dimen6-.5\wd1
                \kern\dimen6
                $\vcenter to\proofdotnumber\proofdotseparation
                        {\leaders\box1\vfill}$%
                \unhbox\proofrulename}%
\else   \dimen6=\fontdimen22\the\textfont2 % height of maths axis
        \dimen7=\dimen6
        \advance\dimen6by.5\proofrulebreadth
        \advance\dimen7by-.5\proofrulebreadth
        \setbox0=\hbox{%
                \kern\shortenproofleft
                \ifdoubleproof
                \then   \hbox to\dimen0{%
                        $\mathsurround0pt\mathord=\mkern-6mu%
                        \cleaders\hbox{$\mkern-2mu=\mkern-2mu$}\hfill
                        \mkern-6mu\mathord=$}%
                \else   \vrule height\dimen6 depth-\dimen7 width\dimen0
                \fi
                \unhbox\proofrulename}%
        \ht0=\dimen6 \dp0=-\dimen7
\fi
\let\doll\relax
\ifwasinsideprooftree
\then   \let\VBOX\vbox
\else   \ifmmode\else$\let\doll=$\fi
        \let\VBOX\vcenter
\fi
\VBOX   {\baselineskip\proofrulebaseline \lineskip.2ex
        \expandafter\lineskiplimit\ifproofdots0ex\else-0.6ex\fi
        \hbox   spread\dimen5   {\hfi\unhbox\proofabove\hfi}%
        \hbox{\box0}%
        \hbox   {\kern\dimen2 \box\proofbelow}}\doll%
\global\dimen2=\dimen2
\global\dimen3=\dimen3
\egroup % NESTED ZERO
\ifonleftofproofrule
\then   \shortenproofleft=\dimen2
\fi
\shortenproofright=\dimen3
\onleftofproofrulefalse
\ifinsideprooftree
\then   \hskip.5em plus 1fil \penalty2
\fi
}

\newcommand{\gammainf}{\ensuremath{\Gamma_\infty}}

\def\SS{{\cal S}}
\newcommand\conv{{\rm conv}}
\newcommand\sort{{\rm sort}}
\newcommand\var{{\rm var}}
\newcommand\app{{\rm app}}
\newcommand\dom{{\rm dom}}
\newcommand{\type}{{\rm type}}
\newcommand\FV{{\rm FV}}
\newcommand\hfv{{\rm hfv}}
\newcommand\hfvt{{\rm hfvT}}
\newcommand\A{{\cal A}}
\newcommand\R{{\cal R}}
\newcommand\V{{\cal V}}
\newcommand\X{{\cal X}}
\newcommand{\oftype}{{:}}
\newcommand{\arr}{{\rightarrow}}
\newcommand{\redb}{\rightarrow_{\beta}}
\newcommand{\eqb}{=_{\beta}}
\newcommand{\Term}{{\rm Term}}

\newcommand\T{{\cal T}}

\newcommand{\merge}{\ltimes}
\newcommand{\comp}{\mathrel{{|}{|}{|}}}

\newcommand{\ocamlcode}[1]{\texttt{#1}}

\newcommand{\realquote}[1]{{``{#1}''}}%%%\emph
\newcommand{\scarequote}[1]{{`{#1}'}}%%%\emph

\usepackage{color}
\definecolor{red}{rgb}{1.0,0.0,0.0}

\newcommand{\NotForPublication}[1]{}%%%change to {} for publication!!!

\begin{document}

\maketitle

%%%
%%%
%%%

\begin{abstract}
  We present an approach to type theory in which the typing judgments
  do not have explicit contexts.  Instead of judgments of shape
  $\Gamma \vdash A : B$, our systems just have judgments of shape 
  $A : B$.  A key feature is that we distinguish free and bound variables
  even in pseudo-terms.

  Specifically we give the rules of the \scarequote{Pure Type System}
  class of type theories in this style.  We prove that the typing
  judgments of these systems correspond in a natural way with those of
  Pure Type Systems as traditionally formulated.  I.e., our systems
  have exactly the same well-typed terms as traditional presentations
  of type theory.

  Our system can be seen as a type theory in which all type judgments
  share an identical, infinite, typing context that has infinitely
  many variables for each possible type.  For this reason we call our
  system \gammainf.  This name means to suggest that our type judgment
  $A : B$ should be read as $\gammainf \vdash A : B$, with a fixed
  infinite type context called \gammainf.
\end{abstract}

\section{Introduction}

\def\der#1#2#3{\setbox0=\hbox{$\strut #2$}%
\hbox to\wd0{$\hss\displaystyle{\strut{#1}\over\strut\box0}\rlap{\raise1pt\hbox{$\,#3$}}\hss$}}
\def\leaf#1{\displaystyle{\strut\atop\strut #1}}

\subsection{Problem}

One of the important insights type theory 
   gives 
us is the need to be aware of the \emph{context}
in which one 
   works. 
This was already stressed by de Bruijn in his 1979 
paper \emph{Wees contextbewust in WOT} \cite{bru:79}, 
Dutch for \realquote{Be context aware in the mathematical vernacular}.
In type theory a term always is considered with respect to a context $\Gamma$, 
   which gives 
the types of the variables occurring free in the term.
This is also apparent in 
   the shape $\Gamma \vdash M : A$ of 
the 
   judgments 
of type theory, 
where the context $\Gamma$ is made explicit.
   Thus 
a \scarequote{bound} variable is bound \emph{locally} in a term, 
while a \scarequote{free} variable actually is \emph{globally} bound, 
namely by the context. %%% of the term.

In 
   customary presentations 
of first order predicate logic 
\cite[for example]{mendelson:intro,vanDalen:LandS}, 
and in fact
in the presentation of most other logics as well, free variables are not
treated in such a way.
In these logics free variables are really \emph{free}.
They are taken from an infinite supply of variables that are just
available to be used in formulas and terms,
without them having to be declared first.

This difference between type theory and predicate logic means that when
we model predicate logic in type theory, actually we do not get
the customary version of predicate logic, but instead get a version
called \emph{free} logic~\cite{lambert:free-logic1963}.
In traditional treatments  
the formula 
\((\forall x.\, P(x)) \to (\exists x.\, P(x))\)
is usually provable.
For instance a natural deduction proof of this formula would look like:
\vspace{-\bigskipamount}
\[
\der{
\der{
\der{
\leaf{[\forall x.\, P(x)]}
}
{P(y)}{\forall E}
}
{\exists x.\, P(x)}{\exists I}
}
{(\forall x.\, P(x)) \to (\exists x.\, P(x))}{{\to}I}
\]
   This proof uses 
the free variable $y$.
If one cannot use any other variables than 
   those 
introduced
by earlier rules, then this proof 
   fails. 
   Indeed, 
the type corresponding to the formula 
   is not 
inhabited: there is no term \(M\) such that 
   the following judgment is derivable: 
\[D : {*},\, P : D \to {*} \vdash M : (\Pi x : D.\, P(x)) \to (\Sigma x : D.\, P(x))\]
because we cannot 
   avoid the case in which the domain \(D\) is empty, 
   where 
the formula is false.

Now there are two things one can do to 
%%%bridge this difference 
   bridge this gap 
between type theory and traditional logic:
\begin{itemize}
\item
Make predicate logic more like type theory, by explicitly keeping
track in the judgments of the set of variables that can be used in
the proof.

\item
Make type theory more like predicate logic, by having a version
of type theory that does not need contexts in the judgments,
i.e., in which free variables are just taken from 
%%%an infinite supply of variables. 
   some infinite supply. 

\end{itemize}

\noindent
Although the first option is interesting too, especially in categorical treatments of logic \cite[for example]{lambek-scott:hocl1988}, in this paper we
focus on the second.
We originally thought that the dependent types in type theory would
prevent a version of type theory without contexts
from being a viable option, but to our surprise it turns out that
one \emph{can} present type
theory in a style where there are no contexts and in which therefore
free variables are really free, provided we are prepared to pay the small price
of labelling variables in a rather involved manner. 

%%%In type theories that are actually implemented in interactive theorem provers, 
   In those type theories actually implemented in interactive theorem provers,
the context always consists of a part holding global \emph{definitions} and parameters and
a part holding the \emph{free variables} in the term 
(as in \cite[for example]{CSC:MKM2004,severi-poll:DPTS}).
For simplicity of exposition, 
and for the sake of proving an exact correspondence between 
a standard presentation of type theory and the variant 
%%%proposed in this paper, we consider here 
   we propose, in this paper we consider 
only the second part of such contexts.
We believe, however, 
that the other part can be treated in exactly the same way.

There is another reason why it is interesting to look at
a version of type theory where there are no explicit contexts.
One of the most popular architectures for proof assistants is the
\emph{LCF architecture}, named after the LCF system from the seventies \cite{LCF-book:1979}.
In the original form of such a system there is an abstract data type called \ocamlcode{term},
whose elements can only be created by a small number of functions exported
from the type-checking kernel.
Elements of this datatype always correspond to type-correct terms,
and those terms can contain free variables.

A system using this approach has 
   a kernel interface containing 
a function:
\begin{quote}
\begin{alltt}
app : term * term -> term
\end{alltt}
\end{quote}
\noindent
When this function is called, the kernel of the system makes sure that
the types of the arguments are compatible, i.e., that the result is 
again a type-correct term.

This is how the HOL family of theorem provers
is implemented.
These systems have a logical foundation that is based on a typed
lambda calculus.
However, in these systems the free variables in the terms are not recorded in
a context of variables.
The only context in these systems is the context of {definitions},
which is kept track of in a stateful variable.
Definitions are never allowed to be removed from this context,
as that would compromise the safety of the kernel.
Hence, these systems are stateful, although they can be made functional using a variant
of the approach presented here \cite{stateless}.

\newpage 
\noindent 
There are two classes of systems that can be said to implement a \emph{type theory},
%%%a typed lambda calculus. These two classes are: 
   a typed lambda calculus: 
\begin{itemize}
\item
The simpler type theories, in which no dependent types are allowed.
They often are a form of simple type theory with some enhancements,
   such as some 
form of polymorphism or type classes.
These include the systems from the HOL family: HOL4, HOL Light, ProofPower and
Isabelle.
These systems can be, and are, implemented following the LCF architecture 
just outlined.
In these systems variables come from an infinite \scarequote{sea} of free variables,
and in the logical theory 
   there is \emph{no} context keeping track of the variables. 

\item
   More  
advanced type theories, often called type theory with \emph{dependent types}.
These 
come from the Dutch \textsc{Automath} systems, the Swedish tradition
of Martin-L\"of type theory, the French tradition of the Calculus of
Constructions and variations on the Edinburgh Logical Framework. 
   Their implementations 
include Coq, NuPRL, Twelf, Agda, Lego, Plastic and Epigram.
In the logical theory of these systems 
   there \emph{is} a context keeping track of the variables. 
\end{itemize}
For this second style of type theory the pure LCF approach is not attractive.
The \ocamlcode{app} function will need to check whether the contexts in which
the terms live are compatible, which will be very expensive, if it needs
to be done for the type-checking of each function application.

For this reason actual type-theoretical proof assistants have a kernel
with a different kind of interface.
In such a kernel there is no abstract datatype of \emph{terms} (there
just is a---non-abstract---type of \emph{pseudo-terms}).
Instead there is an abstract type of well-typed \emph{contexts}, which we call
\ocamlcode{context} here.
(There is also a type \ocamlcode{pseudocontext}, but that is irrelevant here.)
The interface then looks like:
\begin{quote}
\begin{alltt}
mkApp : pseudoterm * pseudoterm -> pseudoterm
add_constant : string * pseudoterm -> context -> context
\end{alltt}
\end{quote}
where \ocamlcode{mkApp} is just the constructor of the data type of
pseudoterms, whereas \ocamlcode{add\_constant} is a function that does
the type-checking: a pseudoterm will only be added to the context
after it has been type-checked.  
%%%(The names of these functions are the actual names they have in the Coq system's kernel.  
   (These are the actual names of the functions in the kernel of the Coq system. 
The types of those functions in 
%%%the Coq system 
   Coq 
are essentially what is presented here.
The type \ocamlcode{pseudoterm} is called \ocamlcode{constr} %%%, 
    in Coq, 
while the type \ocamlcode{context} is called
\ocamlcode{safe\char`\_environment}.)  The system also has a global
variable
\begin{quote}
\begin{alltt}
global_env : context ref
\end{alltt}
\end{quote}
corresponding to the \emph{state} of the system in which the user 
   works. 
%%%This variable 
   It 
is not part of the kernel (and in fact is changed back
   by an undo operation), 
but as there are no
interesting operations combining two different \ocamlcode{context}s,
only one global \ocamlcode{context} is 
%%%relevant at any time, 
   ever relevant, 
%%%which is 
the one given by 
   the contents of 
this variable.

Although the architecture with contexts that we described
is purely functional (as is the Coq kernel), the fact that the actual
implementation has this global variable means that it is used in
a rather \scarequote{stateful} way.
The desire to investigate 
   a possible 
LCF-style kernel for
type theory that is \scarequote{less stateful} motivated this research.
In the conclusions we will address the question whether the style
of type theory we present here will lead to such a type-checking 
architecture.

\subsection{Approach}

The approach 
we will follow here is to imagine there to
be an \scarequote{infinite context} called \gammainf.
For each type-correct type $A$ this context will have infinitely
many variables $x_i^A$.
It should be stressed that this $A$ should be considered to just
be a label, a \emph{string}.
   Reduction will never happen  
inside these $A$s.
Also, $x_i^A$ and $x_i^B$ will be \emph{different} variables, even
when $A$ and $B$ are convertible, or even if they are $\alpha$-equivalent.
Note that the (free) variables in $A$ themselves will also be of shape
$y_j^B$: this means that 
   the variables themselves, as well as the  terms, 
have a recursive
tree-like structure. 

\noindent
For example, 
a variable corresponding to the successor function on natural numbers 
%%%will be something like: 
   looks like: 
$$s^{N^* \to N^*}$$
If we use numbered names for the variables, this might become:
$$x_0^{x_0^* \to x_0^*}$$
So the \realquote{small price} alluded to above is that a free variable $x_i^A$ in a well-typed term 
   carries 
with it the (well-founded) history of how it comes to be well-typed; 
   that is, the 
label $A$ 
   witnesses 
the validity of the context extension $\Gamma, x_i : A$.
 
Now our systems will have judgments $A : B$, which 
should be interpreted as $\gammainf \vdash A : B$.
For this reason we call the general approach to type theory 
   introduced here 
\scarequote{\gammainf} (reusing the name of the context
as the name of the system).
Note that \gammainf\ is not a single system: each type theory
will have a \scarequote{\gammainf-variant}.

The \gammainf\ approach has the essential feature that there
are two different classes of variables.
There are the variables that come from the \gammainf\ context
(the \scarequote{free} variables),
and then there is another kind of \scarequote{bound} variables.
When thinking about our systems one might imagine de Bruijn indices for
the bound variables, although the presentation 
   we give here uses 
named variables for them as well.

Although we expect 
%%%that many type theories will 
   many type theories to  
have a \gammainf-variant, 
   here we only consider  
the class
of type theories called Pure Type Systems \cite{bar-lics-hbk} (PTSs). 
That way 
   we keep everything 
concrete, and 
   it allows 
us to prove a precise correspondence between 
%%%Pure Type Systems 
   PTSs 
in the traditional style and our version 
   in \gammainf\ style. 

One should note that in 
   \gammainf\ 
any type will be inhabited, simply
because there are free variables of every type: in particular, just as in traditional treatments of logic, all domains are assumed to be inhabited. This is in essence the
same as in the version with contexts, except that 
   such variables are there explicitly recorded in the context. 

\subsection{Related Work}

To our surprise, we 
    found little 
published work investigating 
   such an approach to \emph{dependent} type theory.  
   In Church's original formulation of \emph{simple} type theory
   \cite{church:stt1940}, variables, both free \emph{and} bound, are
   annotated with their types, writing for example $\lambda
   x^{\alpha}. f^{\alpha\arr\alpha}\, x^{\alpha}$. (whereas in our formulation
   we would write $\lambda x \oftype {\alpha^*}. f^{\alpha^* \arr\alpha^*}\,
   x$.) Girard adopted \scarequote{Church-style} in the account of System F in his
   thesis \cite{girard:thesis}.
   %%%
   In neither system do \emph{term} variables occur in types, while
   \emph{type} variables are not regulated by an explicit context. In
   these non-dependent type theories, contexts are not strictly
   needed, because one can define the different 
   syntactic classes --- types, terms --- in stages.
   %%% 
   One can regard our approach as extending that of Church to
   dependent types, but optimised to avoid the need to consider
   substitution in labels on \emph{bound} variables which otherwise
   might arise in the application rule.

Conor McBride (private communication) observed that Pollack's \textsc{LEGO} 
implementation already supported the \gammainf\ idea, and this idea was then used 
in his \textsc{OLEG} extensions, and subsequently 
in the architecture of \textsc{Epigram 1}. 
However, this approach has not been treated theoretically as we do here. 

The explicit distinction between free and bound variables on a syntactic level
already can be found in \cite{mckinna-pollack:jar1999}. 
The motivation there was to avoid capture during substitution while keeping a close correspondence with the informal presentation of PTSs with named variables, rather than how to give a \scarequote{\gammainf} presentation, as here. 
Various approaches to representing 
   binding 
are discussed in \cite{Vaughan06areview}, 
   which considers 
named free variables and de Bruijn index bound variables 
one of the best options for mechanisation. 
Indeed, in ongoing work \cite{krebbers:gammainf} the second author has formalised 
one half of the correspondence proved in Section~\ref{sec:thms} in such a style. 
%%%For the sake of presentation, we leave the niceties of 
%%%such a formalisation to the enthusiastic reader. 
   We expand on the niceties of this formalisation 
   in Section \ref{sec:conclusion}.
   %%%below. 

Elsewhere, Pollack considered presentations of type theory separating
the typing judgment from that for context
well-formedness~\cite{randy:closure-under1993}, although judgments are
still \scarequote{in context}. This allows a subtle range of issues to be
explored, especially
   regarding 
closure under \(\alpha\)-conversion, 
%%%which we treat only informally here. 
   here treated only informally. 

Most significantly, but starting from a rather different point,
Sacerdoti Coen considered the problem of proof-checking in the setting
of a distributed library, and hence the problem of how to
\emph{reconstruct} a context in which a given term may be successfully
type-checked~\cite{CSC:MKM2004}. This work (elaborated in
\cite{CSC:PhD}, and forming the basis of the Matita system) goes
beyond the standard PTS setting considered in this paper. 
It identifies a subtle problem which arises when attempting to merge
contexts (including definitions) in the presence of global constraints
(such as universe levels). 

\textit{Added in proof.} Between acceptance of the final version of this paper and this final version, our attention was drawn to the recent work of Matthias Boespflug~\cite{boespflug}, which itself references an earlier (2009) account of our ideas. By contrast with our presentation, which is purely first-order, Boespflug uses higher-order abstract syntax (HOAS), with a view to implementation. 

\subsection{Contribution}

We present a different approach to type theory, 
much closer to the way logical systems
usually are presented than the standard presentation, 
in which free variables are
not bound in a finite context but are taken to be really free.

We 
   validate our approach 
by proving two theorems, 
\ref{thm:PTStoGinf} and \ref{thm:GinftoPTS} below, 
   establishing 
a straightforward correspondence between 
   the standard presentation 
and the variant 
   presented here.

\subsection{Outline}

The structure of the paper is as follows.
In Section~\ref{sec:pts} we recall the
PTS rules and some of its theory.
In Section~\ref{sec:Gi} we present the \gammainf-variant of the
PTS rules, in which the judgments do not have contexts.
In Section~\ref{sec:thms} we show that both systems
correspond to each other in a natural way.
In Section~\ref{sec:conclusion} we conclude with a prospectus for
an implementation based on our variant of the PTS rules.

\section{Pure Type Systems in the traditional style}\label{sec:pts}

Pure Type Systems (PTSs) 
%%%are a generalisation of 
   generalize 
many existing type
systems and thus the class of PTSs contains various well-known
systems, like 
%%%system F and system F$\omega$, 
   systems F and F$\omega$, 
dependent type theory $\lambda$P 
and the Calculus of Constructions.

\begin{definition}\label{def.pts} For $\SS$ a set (the set of \emph{sorts}), 
$\A \subseteq \SS\times\SS$ (the set of \emph{axioms}) and 
$\R \subseteq \SS\times \SS \times \SS$ (the set of \emph{rules}), 
the \emph{Pure Type System} $\lambda(\SS ,\A
,\R)$ is the typed
$\lambda$-calculus with inference rules as in 
Figure~\ref{fig:pts}. 
\begin{figure}[t]
{\normalsize
$$\begin{array}{|llr|}
\hline
%\hoog{1.5em}
(\mbox{sort})
&
\prooftree
\justifies
\vdash s_1 : s_2
\endprooftree &\mbox{if }(s_1 ,s_2)\in \A\\
&&\\
(\mbox{var})
&
\prooftree
\Gamma \vdash A : s
\justifies
\Gamma ,x\oftype A\vdash x : A
\endprooftree &\mbox{if }x\notin\Gamma \\
&&\\
(\mbox{weak})
&
\prooftree
\Gamma \vdash A : s\;\;\;\;
\Gamma \vdash M : C
\justifies
\Gamma ,x\oftype A \vdash M : C
\endprooftree &\mbox{if }x\notin\Gamma \\
&&\\
(\Pi)
&
\prooftree
\Gamma \vdash A : s_1\;\;\;\;
\Gamma ,x\oftype A \vdash B : s_2
\justifies
\Gamma \vdash \Pi x\oftype A .B : s_3
\endprooftree & \mbox{if }(s_1 ,s_2 ,s_3) \in \R \\
&&\\
(\lambda)
&
\prooftree
\Gamma ,x\oftype A \vdash M : B\;\;\;\;
\Gamma \vdash \Pi x\oftype A .B : s
\justifies
\Gamma  \vdash \lambda x\oftype A .M : \Pi x\oftype A.B
\endprooftree &\\
&&\\
(\mbox{app})
&
\prooftree
\Gamma \vdash M : \Pi x\oftype A.B\;\;\;\;
\Gamma \vdash N : A
\justifies
\Gamma  \vdash MN : B[x := N]
\endprooftree &\\
&&\\
(\mbox{conv})
&
\prooftree
\Gamma \vdash M : A\;\;\;\;
\Gamma \vdash B : s
\justifies
\Gamma  \vdash M : B
\endprooftree & A =_{\beta} B\\[1em]
\hline
\end{array}$$
}
\caption{Typing rules for PTSs\label{fig:pts}}
\end{figure}
In the 
rules, the expressions \(M,N,A,B,C\) are taken from the set of \emph{pseudo-terms} 
$\T$ defined by\label{pts-pseudo-terms}
$$ \T ::=
s \bar \V \bar \Pi \V \oftype \T .\T \bar \lambda \V \oftype \T .\T \bar \T \T .$$
with $\V$ 
   a set of variables, 
and the $\Gamma$ taken from the set of \emph{pseudo-contexts}
   \[x_1 : A_1, \ldots , x_n :A_n \quad (x_i \in \V, A_i\in\T, 1\le i\le n)\]
   with the $x_i$ all distinct. 
(We leave the choice of variable \emph{names} 
unspecified at this point, 
as this does not matter as long as 
   \(\V\) is countably infinite, but  
below we will take a specific choice 
   of names.) 
\end{definition}

There is a lot of theory about PTSs and various systems have been
studied in the context of PTSs. We do not give a complete overview but
refer to \cite{bar-lics-hbk,bargeu-ar-hbk,geuv93} for details and
explanation. Here we just give the results that we 
use in the rest of the paper to prove the equivalence between 
a PTS and its
   \gammainf-variant. 

\begin{definition}
The pseudo-term $A$ is called \emph{well-typed} if a pseudo-context $\Gamma$ 
and pseudo-term $B$ exist such that $\Gamma\vdash A:B$ or
$\Gamma\vdash B:A$ is derivable. A pseudo-context $\Gamma$ is 
\emph{well-formed} if pseudo-terms $A$ and $B$ exist such that
$\Gamma\vdash A:B$ is derivable; a \emph{context}  is 
   a well-formed pseudo-context. 
The set of variables declared in 
pseudo-context $\Gamma$ is called the \emph{domain} of $\Gamma$,
$\dom(\Gamma)$. For $x\in\dom(\Gamma)$, let $\type_{\Gamma}(x)$
denote the \scarequote{type} assigned to $x$ in $\Gamma$: if $x:A
\in\Gamma$, then $\type_{\Gamma}(x) = A$. The expression $\type(\Gamma)$ denotes the
set of such \scarequote{types} occurring in $\Gamma$. The set of well-typed terms of
$\lambda (\SS,\A,\R)$ is denoted by $\Term(\lambda (\SS,\A,\R))$.
\end{definition}

We adopt the usual notions of bound and free variable, 
$\alpha$-conversion ($\equiv$), 
   substitution (\(B[x := N]\), used in the rule (\mbox{app})), 
$\beta$-reduction ($\redb$) and
$\beta$-equality ($\eqb$, used in the rule (\mbox{conv})) on pseudo-terms. 

The following are well-known properties of PTSs. 
The relation $\Gamma\subseteq\Delta$ 
%%%just denotes set containment 
   denotes inclusion  
between $\Gamma$ and $\Delta$ 
regarded as sets of variable assignments. 
%%%The third property, 
   The third, 
Permutation, is a corollary of Strengthening.

\begin{proposition}\label{prop.metatheoryPTS}
\ 
\begin{description}
\item[Thinning] If $\Gamma \vdash M:A$ and $\Delta \supseteq \Gamma$ is
   well-formed, 
then $\Delta \vdash M:A$.
\item[Strengthening] If $\Gamma, x:B,\Delta \vdash M:A$ and
$x\notin\FV(\type(\Delta),M,A)$, 
then $\Gamma, \Delta \vdash M:A$.
\item[Permutation] If $\Gamma, x:B,y:C,\Delta \vdash M:A$ and
$x\notin\FV(C)$,  
then $\Gamma, y:C, x:B,\Delta \vdash M:A$.
\end{description}
\end{proposition}

In proving the equivalence between a PTS and its 
   \gammainf-variant, 
we need to merge two 
contexts to create a new one. Therefore
we introduce 
   the following:  

\begin{definition}\label{def.comp}
Let $\Gamma$ and $\Delta$ be two pseudo-contexts. 
%%%We say that 
   We say 
$\Gamma$ and $\Delta$ are \emph{compatible}, notation $\Gamma\comp\Delta$, if
$$\forall x\in \dom(\Gamma)\cap\dom(\Delta)( \type_{\Gamma}(x) \equiv \type_{\Delta}(x)).$$
The \emph{merge} of $\Gamma$ and $\Delta$, notation
$\Gamma\merge\Delta$, is the pseudo-context $\Gamma ,
(\Delta\setminus\Gamma)$. This is $\Gamma$ followed by the
declarations $x:B\in\Delta$ for which $x\notin\dom(\Gamma)$.
\end{definition}

Note the strong requirement in $\Gamma\comp\Delta$ that the types of
$x$ in $\Gamma$ and $\Delta$ should be $\alpha$-equal, and not just
$\beta$-convertible.

\begin{lemma}
If $\Gamma$ and $\Delta$ are 
contexts and $\Gamma \comp \Delta$, then $\Gamma \merge\Delta$ is well-formed.
\end{lemma}

\begin{proof}
We write $x_1:B_1, \ldots, x_n:B_n$ for $\Delta$.  $\Gamma \merge\Delta$
is the pseudo-context $\Gamma, (\Delta\setminus\Gamma)$.  As $\Gamma$
is well-formed, we only have to consider the part  
$\Delta\setminus\Gamma = x_{i_1} : B_{i_1}, \ldots, x_{i_n} : B_{i_n}$. 
%%%\\
   But 
$x_1:B_1, \ldots, x_{i_1 -1}: B_{i_1 -1} \subseteq \Gamma$, so by Thinning $\Gamma \vdash B_{i_1} :s$ for some sort $s$, so $\Gamma, x_{i_1} :B_{i_1}$ is well-formed.\\
The same reasoning applies to $x_{i_2}:B_{i_2}, \ldots, x_{i_n}:B_{i_n}$, 
so we conclude that $\Gamma, (\Delta\setminus\Gamma)$ is well-formed. 
\end{proof}

%%%%%%%%%%%%%%%%
In our system \gammainf, the free variables will have special names,
as they are labelled by their types. Of course, consistently renaming
the free variables in a judgment $\Gamma \vdash M:A$ does not change
its meaning, as the free variables in $M$ and $A$ are actually bound
in $\Gamma$. For clarity, we introduce this notion explicitly.

\begin{definition}\label{def:alpha-judg}
The judgment $\Gamma \vdash M:A$ is $\alpha$-equivalent to 
%%%the judgment 
$\Gamma ' \vdash M':A'$ in case one can be obtained from the
other by renaming bound variables, where we consider the free
variables in $M$ and $A$ to be bound by their declaration in $\Gamma$.
\end{definition}
%%%%%%%%%%%%%%%%

\section{Pure Type Systems in the \gammainf\ style}\label{sec:Gi}

We now make the set of 
   variables 
$\V$ explicit.
We have two kinds:  
variables $\dot x$ with a dot on top,  
   intended
to be bound by $\lambda$ and $\Pi$ binders; and variables $x^A$ 
   tagged with a pseudo-term $A$, 
   intended 
to be bound in the context.
This means we take $\V$ in Definition~\ref{def.pts} 
   as follows, where $\X$ supplies the \emph{names} 
   of 
variables: 
$$\begin{array}{rcl}
\V &::=& \dot\X \bar \X^\T \\
\X &::=& x \bar y \bar z \bar \dots \bar x_0 \bar x_1 \bar x_2 \bar \dots
\end{array}$$
Clearly the rules for $\T$ and $\V$ are mutually recursive.

Note that although the variables are \emph{intended} to be used in a certain
way (made precise in Definition~\ref{def.typeannotated} below), 
in the PTSs as defined
   above 
both kinds of $\V$ can be used for all purposes.
In particular $\dot x$ can be put in a context, $x^A$ can be bound, 
and the 
   label 
$A$ of $x^A$ need not correspond to the type of $x^A$.

Note also that the 
   definition of substitution, and hence the 
relation of $\beta$-equality, 
   is agnostic 
about the structure of the annotations of the variables.
(The definition
of $=_\beta$ in Section \ref{sec:pts} takes $\V$ just as a set).
This means that although
$$(\lambda \dot A:{*}. \dot A)\,B^{*} =_\beta B^{*}$$
we have that
$$x^{(\lambda \dot A:{*}. \dot A) B^{*}} \ne_\beta x^{B^{*}}$$

We will now define the rules of \gammainf.
To do this we first have to introduce the notion of \emph{hereditarily free variables of
the types of the free variables} in a term.
We first motivate the need for this notion.

In a PTS, the context takes care that one can only abstract over a
variable if nothing else depends on that variable. In the rules, this
is formalised by requiring that the $x:A$ 
   abstracted 
over in
the $\Pi$ or $\lambda$ rule must be the last declaration in the
context. This ensures that $x$ does not occur free in any of the other
types in the context. In \gammainf\ we do not have contexts, so we have to replace this by another side condition on the rules. We would like to have a $\Pi$ rule as follows:
$$
{
A : s_1
\quad\enskip
B : s_2
\over
\Pi \dot y {:} A.\, B[x^A := \dot y] : s_3
}
\rlap{ \small $(s_1,s_2,s_3)\in{\cal R}$}
$$
but that is wrong, because then we would be able to form (in PTS terminology)
the $\Pi$-type 
$$
{
\ldots \vdash A:* \quad \enskip A:*, P:A\arr *, Q:\Pi x\oftype A. P\, x \arr *,  a:A, h:P\, a \vdash  Q \, a\, h :*
\over
A:*, P:A\arr *, Q:\Pi x\oftype A. P\, x \arr *, h:P\, a \vdash  \Pi y :A.Q\,y\, h : *
}
$$
%%%This is not correct, 
   But this cannot be correct, 
because $h$ is not of type $P\, y$ in the conclusion. 
In \gammainf, this derivation would read (adding some brackets for readability):
$$
{
A^*:* \quad \enskip 
Q^{\Pi \dot x\oftype A^*. (P^{A^*\arr *}\, \dot x) \arr *} \, a^{A^*}\, h^{P^{A^*\arr *}\, a^{A^*}} :*
\over
\Pi \dot y :A^*. Q^{\Pi \dot x\oftype A^*. (P^{A^*\arr *}\, \dot x) \arr *} \, \dot y\, h^{P^{A^*\arr *}\, a^{A^*}} :*
}
$$
which would be 
   derivable 
according to the $\Pi$-rule above, but
clearly 
   undesirable.

\begin{definition}\label{def:HFV}
Given $M\in \T$, we define the \emph{hereditarily free variables in $M$}, denoted $\hfv(M)$, as follows:
\begin{eqnarray*}
\hfv(s) =\hfv(\dot x) &=& \emptyset\\
\hfv(x^{A})&=& \{x^{A}\} \cup\hfv(A)\\
\hfv(F\,N) &=& \hfv(F)\cup\hfv(N)\\
\hfv(\lambda \dot x \oftype A.N) = \hfv(\Pi \dot x\oftype A.N) &=& \hfv(A)\cup\hfv(N)
\end{eqnarray*}
\end{definition}

So, where $=_\beta$ basically ignores the structure of
the type labels of 
free variables, we now take them seriously, 
collecting the variables (hereditarily) free in the type labels as well.

We put as a side condition in the $\Pi$-rule that $x^A$ should
not occur free in any of the \emph{types of the free variables in $B$}, and
similarly for the $\lambda$ rule. We give an explicit %%%direct 
definition of this notion.

\begin{definition}\label{def:HFVT}
Given $M\in \T$, we define the \emph{hereditarily free variables of the types of the free variables in $M$}, denoted $\hfvt(M)$, as follows:
\begin{eqnarray*}
\hfvt(s) =\hfvt(\dot x) &=& \emptyset\\
\hfvt(x^{A})&=& \hfv(A)\\
\hfvt(F\,N) &=& \hfvt(F)\cup\hfvt(N)\\
\hfvt(\lambda \dot x \oftype A.N) = \hfvt(\Pi \dot x\oftype A.N) &=& \hfvt(A)\cup\hfvt(N)
\end{eqnarray*}
\end{definition}

So, for example $\hfvt(h^{P^{A^*\arr *}\, a^{A^*}}) = \{ P^{A^*\arr *} , a^{A^*}, A^*\}$. An easy corollary of the definition is that
$$\hfvt (M) \subseteq \hfv(M)$$

We now give the derivation rules of the system.

\begin{definition}
The derivation rules of \gammainf\ are 
given by the inference rules in Figure~\ref{fig:gammainf}. 
\begin{figure}[t]
{\normalsize
$$
\begin{array}{|llr|}
\hline
(\mbox{\rm sort})
&
\prooftree
\justifies
s_1 : s_2
\endprooftree &\mbox{if }(s_1 ,s_2)\in \A\\
&&\\
(\mbox{\rm var})
&
\prooftree
A : s
\justifies
x^A : A
\endprooftree & \\
&&\\
(\Pi)
&
\prooftree
A : s_1\;\;\;\;
B : s_2
\justifies
\Pi \dot x\oftype A .B[y^A := \dot x] : s_3
\endprooftree & \hspace{-2em}\mbox{if }(s_1 ,s_2 ,s_3) \in \R\mbox{ and }y^A \notin\hfvt(B)\\
&&\\
(\lambda)
&
\prooftree
M : B\;\;\;\;
\Pi \dot x\oftype A .B[y^A := \dot x] : s
\justifies
\lambda \dot x\oftype A .M[y^A := \dot x] : \Pi \dot x\oftype A.B[y^A := \dot x]
\endprooftree & \mbox{if }y^A \notin\hfvt(M)\cup\hfvt(B)\\
&&\\
(\mbox{\rm app})
&
\prooftree
M : \Pi \dot x\oftype A.B\;\;\;\;
N : A
\justifies
MN : B[\dot x := N]
\endprooftree &\\
&&\\
(\mbox{\rm conv})
&
\prooftree
M : A\;\;\;\;
B : s
\justifies
M : B
\endprooftree & A =_\beta B\\[1em]
\hline
\end{array}
$$
}
\caption{Typing rules for \gammainf\label{fig:gammainf}}
\end{figure}
The $\Pi$ and $\lambda$ rules have the side condition that $\dot x$ 
should not be \emph{captured} in $B$ or $M$ when doing the substitution.
That is, $\dot x$ should not be bound by a binder under which $y^A$ occurs.
\end{definition}

The side condition on the $\Pi$ and $\lambda$ rules is no restriction as you can
go to an $\alpha$-equivalent version of the term afterwards.

\section{The correspondence theorems}\label{sec:thms}
We now prove that a PTS and its \gammainf-variant correspond to each
other. For a \gammainf\ judgment $M:A$ we generate a context $\Gamma$
such that $\Gamma\vdash M:A$ is 
   PTS-derivable.  
Conversely, if
$\Gamma \vdash M : A$ is 
   PTS-derivable,  
we always have an
$\alpha$-equivalent judgment (see Definition \ref{def:alpha-judg})
$\Gamma'\vdash M':A'$ such that $M' : A'$ is derivable in \gammainf.
This specific form $\Gamma' \vdash M' : A'$ we call a \emph{type
  annotated} judgment.

\begin{definition}\label{def.typeannotated}
A \emph{type annotated context} in a PTS is a context of the form
$$x^{B_1}_1 : B_1, \ldots, x^{B_n}_n:B_n$$
where we moreover assume that all bound variables in the $B_i$ are of the form $\dot x$.
\end{definition}

\begin{definition}\label{def.typeannotated1}
A \emph{type annotated judgment} in a PTS is 
   one 
of the form
$$x^{B_1}_1 : B_1, \ldots, x^{B_n}_n:B_n \vdash M:A$$
where: $x^{B_1}_1 : B_1, \ldots, x^{B_n}_n:B_n$ is a type annotated context; 
all free variables in
$M$ and $A$ are of the form \(x^{B_i}_i\); 
   and 
all bound variables are of the form $\dot x$.
\end{definition}

We now first show the easy direction of our correspondence result:

\begin{lemma}
Every judgment $\Gamma \vdash M:A$ in a PTS is $\alpha$-equivalent to a type annotated judgment $\Gamma' \vdash M':A'$.
\end{lemma}

\begin{proof}
From left to right we rename the variables in the context (and of course also in $M$ and $A$) to the \scarequote{standard} names
$$x_1^{B_1}:B_n,\dots,x_n^{B_n}:B_n$$
(Here $x_i$ is not a meta-variable for a variable name, but really
the \emph{explicit} variable name \realquote{$x_i$} in $\X$.)
We also $\alpha$-rename any bound variable of the form $x^A$ to a fresh
variable of the form $\dot x$. 
\end{proof}

As an example, consider the PTS judgment
$$A : {*},\, a : A \,\vdash (\lambda x\oftype A.\,x)\, a : A$$
This does not fit the variable names from our $\V$, so this does not
conform to the definitions in this paper.
Instead using our variables it should be something like:
$$\dot A : {*},\, a^{*} : \dot A \,\vdash (\lambda x^{\dot B}\oftype\dot A.\,x^{\dot B})\, a^{*} : \dot A$$
This of course is not a \emph{type annotated} judgment, the annotations 
   make no 
sense at all,
but still this is a perfectly
fine 
   PTS judgment 
as defined in Definition~\ref{def.pts}.

Now according to the theorem this is $\alpha$-equivalent to a judgment that \emph{is} type annotated.
And it is, for example it is $\alpha$-equivalent to
$$x_1^{*} : {*},\, x_2^{x_1^{*}} : x_1^{*} \,\vdash (\lambda \dot x\oftype x_1^{*}.\,\dot x)\, x_2^{x_1^{*}}$$
Or, if one uses more readable names, to
$$A^{*} : {*},\, a^{A^{*}} : A^{*} \,\vdash (\lambda \dot x\oftype A^{*}.\,\dot x)\, a^{A^{*}}$$

The other part of the easy direction of our correspondence is that
a type annotated PTS judgment essentially is the same as a \gammainf\ judgment:

\begin{theorem}\label{thm:PTStoGinf}
If 
   the type annotated PTS judgment $\Gamma \vdash M : A$ is derivable,  
then $M : A$ is a derivable judgment of the corresponding \gammainf\ type theory.
\end{theorem}

\begin{proof}
By induction on the size of the derivation of $\Gamma \vdash M : A$.
We do a case split on the last rule used in the derivation: %%%; we treat a few cases.
\begin{itemize}
\item[(sort)]
Immediate.

\item[(weak)]
Trivial, because if $\Gamma,x\oftype A$ is a type annotated context, then certainly
$\Gamma$ is a type annotated context.

\item[(\var)] 
By induction we have $A : s$ derivable in \gammainf, 
and because the context is type annotated
the variable name must be of the form $x^A$.
Hence $x^A : A$ in \gammainf.

\item[(\conv)] 
We know that $\Gamma \vdash M : A$ and $\Gamma \vdash B : s$.
Now $A$ need not have bound variables of the form $\dot x$, 
but 
one can rename them to obtain $A' \equiv A$ such that $\Gamma \vdash M : A'$ \emph{will} be
type annotated.
Then $M : A'$ and $B : s$ in \gammainf\ by induction and therefore also $M : B$ in \gammainf.

\item[(\app)] 
Again, we might need to change the bound variable in $\Gamma \vdash M : \Pi x\oftype A.B$ to the dotted kind, to get a type annotated judgment.
Apart from that this case is trivial, just like the previous one.

\item[($\Pi$)]
The conclusion will be of the form $\Gamma \vdash \Pi\dot x\oftype A.B : s_3$.
Now if we take $y^A$ a completely fresh variable, then $\Gamma,y^A:A \vdash B[\dot x := y^A] : s_2$ will be a type annotated judgment, as well as being $\alpha$-equivalent to 
$\Gamma,\dot x:A \vdash B : s_2$.
   Accordingly, let 
$B' := B[\dot x := y^A]$, so that $B \equiv B'[y^A := \dot x]$.

Clearly now $y^A \notin\hfvt(B')$ because $\Gamma,y^A:A \vdash B' : s_2$,
so all type annotations will be typable in $\Gamma$, which does not contain
$y^A$.

By induction $A : s_1$ and $B' : s_2$ in \gammainf\ and because
$y^A \notin\hfvt(B')$ we get that
$\Pi\dot x\oftype A.B'[y^A := \dot x] : s_3$ in \gammainf.
But this is precisely $\Pi\dot x\oftype A.B : s_3$.
\item[($\lambda$)]
This case essentially follows that of the previous one.

\end{itemize}
\end{proof}

The other direction of our correspondence---from \gammainf\ to
traditional PTS style---is a bit more involved.  We need to
\emph{synthesise} an appropriate context, and because we build this
context by recursion over the type derivation, we need to merge these
synthesised contexts using $\merge$.  For this we need a number of
lemmas involving type annotated contexts in PTSs.

\begin{lemma}\label{lem.annotctxt}
If $\Gamma$ and $\Delta$ are type annotated contexts, then $\Gamma\comp\Delta$.
\end{lemma}

\begin{proof}
If $x^A \in \dom(\Gamma)\cap\dom(\Delta)$, then $x^A : A\in\Gamma$ and
$x^A :A \in \Delta$, of course, $A \equiv A$ and this is what we need
to prove according to Definition \ref{def.comp}. 
\end{proof}

\begin{lemma}\label{lem.annothfvt}
If $\Gamma$ is a type annotated context with $x^A\in\dom(\Gamma)$, $\Gamma \vdash M:B$ and $x^A \notin\hfvt(M,B)$, then 
$$\exists \Delta \subset \Gamma (\Delta, x^A:A\vdash M:B)$$
\end{lemma}

\begin{proof}
Write $\Gamma =\Gamma_1, x^A :A, \Gamma_2$, and suppose $y^C \in
\dom(\Gamma_2)$ with $x^A\in\FV(C)$. If $y^C \in \FV(M,B)$, then $
x^A\in \hfvt(M,B)$, contradiction. So $y^C \notin\FV(M,B)$. This means
that all declarations $y^C:C$ in $\Gamma_2$ for which $x^A\in\FV(C)$
can be removed by Strengthening (Proposition \ref{prop.metatheoryPTS}),
starting from the rightmost declaration in $\Gamma_2$. We end up with a judgment
$$\Gamma_1, x^A :A, \Gamma'_2 \vdash M:B$$ with $\Gamma'_2 \subseteq
\Gamma_2$ and $x^A\notin\type(\Gamma'_2)$. Using Permutation
(Proposition \ref{prop.metatheoryPTS}), we conclude that $\Gamma_1,
\Gamma'_2, x^A :A \vdash M:B$ and we take $\Gamma_1, \Gamma'_2$ for $\Delta$. 
\end{proof}

\begin{corollary}\label{cor.annothfv}
If $\Gamma \vdash M:B$ is a type annotated judgment, there is a
$\Delta \subseteq \Gamma$ such that $\Delta \vdash M:A$ and
$\dom(\Delta) \subseteq \hfv(M,B)$.
\end{corollary}

\begin{proof}
Let $x^A\in \dom(\Gamma)$ and $x^A\notin \hfv(M,B)$. Then $x^A\notin
\hfvt(M,B)$, so (according to Lemma \ref{lem.annothfvt}), there is a $\Delta\subseteq
\Gamma$ such that $\Delta, x^A:A\vdash M:B$. But also $x^A\notin
\FV(M,B)$, so by Strengthening (Proposition \ref{prop.metatheoryPTS}),
$\Delta\vdash M:B$. 
\end{proof}

So, in $\Gamma \vdash M:B$, we can always make the context $\Gamma$ so
small that its domain is within the set of hereditarily free variables
of $M,B$. The other way around, the hereditarily free variables of
$M,B$ should be in $\dom(\Gamma)$:

\begin{lemma}\label{lem.hfvdom}
If $\Gamma \vdash M:B$ is type annotated, then $\hfv(M,B) \subseteq \dom(\Gamma)$.
\end{lemma}

\begin{proof}
We prove $\Gamma \vdash M:B \Rightarrow \hfv(M) \subseteq
\dom(\Gamma)$, by induction on the derivation and then we are done,
because if $\Gamma \vdash M:B$, then $\Gamma \vdash B:s$ for some sort
$s$, or $B$ is a sort. %%%We treat a few cases.
\begin{itemize}
\item[(\sort)] Immediate. 
\item[(\var)] By induction, $\hfv(A) \subseteq\dom(\Gamma)$, so
  $\hfv(x^A) \subseteq \dom(\Gamma, x^A : A)$.
\item[(\conv)] 
By induction, $\hfv(M)\subseteq\dom(\Gamma)$, so we are done.
\item[(\app)] 
By induction, $\hfv(F) \subseteq\dom(\Gamma)$ and $\hfv(M)\subseteq\dom(\Gamma)$, 
so $\hfv(F\,M) \subseteq\dom(\Gamma)$.  
\item[($\Pi$)]
By induction, $\hfv(A)\subseteq\dom(\Gamma)$ and $\hfv(B)\subseteq\dom(\Gamma, y^A:A)$, so 
$\hfv(\Pi \dot x\oftype A. B[y^A := \dot x]) \subseteq \dom(\Gamma)$. 
\item[($\lambda$)]
By induction, $\hfv(M)\subseteq\dom(\Gamma, y^A:A)$ and 
$\hfv(\Pi \dot x\oftype A. B[y^A := \dot x]) \subseteq \dom(\Gamma)$, so 
\(\hfv(A) \subset \dom(\Gamma)\), and hence 
$\hfv(\lambda \dot x\oftype A. M[y^A := \dot x]) \subseteq \dom(\Gamma)$.

\end{itemize}
\end{proof}

The more difficult direction of our correspondence now follows:

\begin{lemma}
Let $M : A$ be a derivable \gammainf\ judgment.
Then all free variables in $M$ and $A$ have the form $x^A$ and
all bound variables have the form $\dot x$.
\end{lemma}

\begin{proof}
By induction on the derivation of $M : A$. 
\end{proof}

\begin{theorem}\label{thm:GinftoPTS}
Let $M : A$ be a derivable \gammainf\ judgment. 
Then there is a type annotated judgment $\Gamma \vdash M:A$ 
derivable in the associated PTS, 
such that $\Gamma$ contains exactly the
variables in $\hfv(M) \cup \hfv(A)$.
\end{theorem}

\begin{proof}
By induction on the derivation of $M:A$ we show 
there exists
a type annotated context $\Gamma(M,A)$ 
   such that $\Gamma(M,A) \vdash M:A$. 
%%%We treat a few cases.
Note that 
$\Gamma(M,A)$ depends on the \emph{derivation} 
of the judgment $M:A$, not just on the terms 
$M$ and $A$.
\begin{itemize}
\item[(sort)]
Immediate.
\item[(\var)] 
By induction, $\Gamma(A,s) \vdash A:s$. So $\Gamma(A,s), x^A : A \vdash x^A : A$. 
\item[(\conv)] 
By induction, $\Gamma(M,A) \vdash M:A$ and $\Gamma(B,s) \vdash B:s$,
and we also know that $A\eqb B$. \\ So
$\Gamma(M,A)\merge\Gamma(B,s)\vdash M : B$ by Thinning and 
the (conv) rule. 
\item[(\app)] 
By induction, $\Gamma(F,\Pi \dot x\oftype A.B) \vdash F:\Pi \dot x \oftype A. B$ and $\Gamma(M,A) \vdash M:A$. \\ So
$\Gamma(F,\Pi \dot x\oftype A.B)\merge\Gamma(M,A)\vdash F\,M : B[\dot x := M]$ by Thinning and 
   the (app) rule. 
\item[($\Pi$)] By induction, $\Gamma(A,s_1) \vdash A:s_1$ and
  $\Gamma(B,s_2) \vdash B:s_2$.\\ 
If $y^A \notin\Gamma(B,s_2)$, then $\Gamma(A,s_1)\merge
\Gamma(B,s_2), y^A:A_i \vdash B:s_2$, so by
Thinning and the ($\Pi$) rule we have $\Gamma(A,s_1)\merge
\Gamma(B,s_2)\vdash \Pi \dot x\oftype A. B[y^A := \dot x] : s_3$.\\
If $y^A \in\dom(\Gamma(B,s_2))$, then $\Delta , y^A:A \vdash B:s_2$ for
some $\Delta \subset \Gamma(B,s_2)$. So  by
Thinning and the ($\Pi$) rule we have $\Gamma(A,s_1)\merge
\Delta\vdash \Pi \dot x\oftype A. B[y^A := \dot x] : s_3$.

\item[($\lambda$)]
By induction, we obtain 
$\Gamma(M,B) \vdash M:B$   
and $\Gamma(\Pi \dot x\oftype A. B[y^A := \dot x],s) \vdash \Pi \dot x\oftype A.\penalty100 B[y^A := \dot x]:s$.\\ 
If $y^A \notin\dom(\Gamma(M,B))$, then $\Gamma(A,s_1)\merge \Gamma(M,B),
y^A:A_i \vdash M:B$. So by Thinning and 
   the ($\lambda$) rule, we have 
 $\Gamma(\Pi \dot x\oftype A. B[y^A :=
  \dot x],s)\merge \Gamma(M,B)\vdash \lambda \dot x\oftype A. M[y^A :=
  \dot x]: \Pi \dot x\oftype A. B[y^A := \dot x]$. \\
If $y^A \in\dom(\Gamma(M,B))$, then $\Delta , y^A:A \vdash M:B$ for some
$\Delta \subset \Gamma(M,B)$, by Lemma \ref{lem.annothfvt}. So by Thinning and 
   the ($\lambda$) rule, we have 
$\Gamma(\Pi \dot x\oftype A. B[y^A :=
  \dot x],s)\merge \Delta\vdash \lambda \dot x\oftype A. M[y^A := \dot x]:
\Pi \dot x\oftype A. B[y^A \penalty 100 := \dot x]. $ 
\end{itemize}
By Lemma \ref{lem.hfvdom}, 
   $\dom(\Gamma(M,A)) \supseteq \hfv(M,A)$. 
Corollary \ref{cor.annothfv} lets us strengthen the
context to $\Delta \subseteq \Gamma(M,A)$ such that $\Delta \vdash
M:A$ and $\dom(\Delta) =\hfv(M,A)$.  
\end{proof}

\begin{corollary}
Let
$$M : A$$
be a derivable \gammainf\ judgment.
Take all variables of the form $x^A_i$ occurring
in $\hfv(M) \cup \hfv(A)$, and put them in \emph{any} order
$x^{A_1}_{i_1}, \dots, x^{A_n}_{i_n}$
such that if $x^{A_k}_{i_k}$ occurs in $A_l$ then $k < l$. 
Then the following judgment is derivable in the PTS: 
$$x^{A_1}_{i_1} : A_1, \dots, x^{A_n}_{i_n} : A_n \vdash M : A$$
\end{corollary}

\begin{proof} 
From the previous Theorem using Permutation. 
\end{proof}

%%%\section{Formalising the correspondence}\label{sec:formalisation}

\section{Conclusion and Further work}\label{sec:conclusion}

There are three obvious continuations of this work:
\begin{enumerate}
\item
The first is to investigate to what extent other type
theories than the PTSs admit a \gammainf\ presentation.

\item
The second is to see how well the approach presented here
can be used as a basis of an LCF-style kernel for type theory.

\item The third is to formally develop the theory presented in
	this paper in a proof assistant.
\end{enumerate}
\noindent
With respect to the first point: we expect most type theories to have
a \gammainf-variant,  
   although the observations about universe inconsistency \cite{CSC:MKM2004}
   arising from merging contexts may complicate the picture for applied type systems 
   such as that of Coq. 
More important is to
investigate how our approach needs to be adapted to support type
  theories with definitions. 
As previously noted, in any real implementation, 
the definitions for defined constants
   form 
a more significant part of the contexts $\Gamma$ we are
   eliminating 
than the free variables. 

   We are currently investigating the second point, developing an
OCaml implementation for the PTS $\lambda P$ extended with definitions
(a system corresponding to the logical framework LF) along the
lines of this paper.  
%%% 
The main question 
is how expensive, computationally, the two following operations are:
\begin{itemize}
\item
The substitutions $[y^A := \dot x]$ that occur in the $\lambda$
and $\Pi$ rules.

\item
The check of the side-condition $y^A \notin\hfvt(M,B)$ in the
$\lambda$ and $\Pi$ rules.

\end{itemize}
\noindent
The first is in some sense \scarequote{local}, because it does not look inside
   constant definitions 
in the environment.  To make the second %%%check
reasonably efficient will be harder.  It 
   is possible we
need to consider \emph{three} kinds of variables, 
   distinguishing
(1) $\dot x$ bound variables, (2) $y^A$ variables 
   to be
substituted with bound variables later (essentially, the \emph{eigenvariables} of the (\(\Pi\)) and (\(\lambda\)) rules), and (3) $x^A$ variables that 
   do 
remain free, so they may be 
   considered as 
\scarequote{axiomatic constants} of the
system. 

In such a system, it is essential that the implementation
language can use pointer equality to efficiently determine equality of
terms (and in particular, equality of annotations in variable
occurrences).  This motivated our choice of OCaml, which is such a
language.  Although in OCaml the comparison function \realquote{\ocamlcode{=}}
does not have this feature (because floating points NaNs are not 
taken to be equal
to themselves, the system never looks at pointer equality when
evaluating \realquote{\ocamlcode{x = y}}), the comparison function
\realquote{\ocamlcode{fun x y -> Pervasives.compare x y = 0}} does.

We are currently working on the third point as well. In ongoing work \cite{krebbers:gammainf} 
the second author has formalised a large part of the theory presented in this paper 
up to the first direction of the correspondence theorem in the proof assistant 
Coq. However, the presentation in this formal development 
is slightly different from that of this paper. 
\begin{itemize}
\item Firstly, we distinguish bound from free variables at the level
        of PTS pseudo-terms, following existing practice,
        established since the third author's work with Pollack in the
        1990s~\cite{mckinna-pollack:tlca1993,mckinna-pollack:jar1999}. 
        Our informal presentation
        above uses named variables in each case, the so-called
        \textit{locally named} approach, whereas the formalisation uses
        %%%More specifically, we have used the locally nameless representation: 
           the \textit{locally nameless} representation: 
        %%%
        de Bruijn indices for bound variables and 
	names for free variables. Because we make this difference at the level of 
	PTS pseudo-terms already, we have a canonical representative for each term and 
	therefore need not worry about $\alpha$-equivalence. 
        For further details of both approaches, see 
        \cite[for example]{mckinna-pollack:jar1999, aydemir:metatheory}. 
 
\item Secondly, variable binding in the ($\Pi$) and ($\lambda$) rules of 
	\gammainf\ is handled by substituting free variables for bound variables, 
        rather than bound for free as in Section~\ref{sec:Gi}. 
	%%%The choice for this presentation is motivated by the fact that it has been used 
	%%%successfully for other formalisations before. 
           This choice has been used successfully in other formalisations
           (\textit{ibid.}), and emphasises the conceptual priority of
           free variables over bound. 
           Recent work by Pollack and Sato compares the two approaches, 
           in a detailed account of canonical representation for languages 
           with binding~\cite{SatoPollack:JSC2010}. 

\end{itemize}

	Based on the methodology described in \cite{aydemir:metatheory},
	we have combined the locally nameless presentation 
        with co-finite quantification to obtain strong
	induction principles. To be sure that our co-finite presentation is
	adequate we have proved it to be equivalent to an exists-fresh presentation.

For example the ($\Pi$) rule, in respectively exists-fresh and co-finite presentation, 
is as follows.
$$
\begin{array}{llr}
(\Pi\mbox{ exists-fresh})
&
\prooftree
A : s_1\;\;\;\;
B[0:=x^A] : s_2
\justifies
\Pi A .B : s_3
\endprooftree & \mbox{if }(s_1 ,s_2 ,s_3) \in \R\mbox{ and }x^A \notin\hfv(B)\\
& & \\
(\Pi\mbox{ co-finite})
&
\prooftree
A : s_1\;\;\;\;
\forall x \notin L\ .\ B[0:=x^A] : s_2
\justifies
\Pi A .B : s_3
\endprooftree & \mbox{if }(s_1 ,s_2 ,s_3) \in \R\mbox{ and }L \subset_{\mbox{finite}} \V\\
\end{array}
$$
Observe that we require $x^A \notin\hfv(B)$ instead of $x^A \notin\hfvt(B)$ 
   (this was also observed by John Boyland).
The reason for this is that we have to bind every occurrence of the free variable 
$x^A$ in $B[0:=x^A]$, hence $x^A$ should not be in $\FV(B)$ either. While the condition is (potentially) more expensive to check, it removes the need for Definition~\ref{def:HFVT}, in favour of the (conceptually) simpler Definition~\ref{def:HFV}. 

The other direction of the correspondence theorem 
%%%requires a proof of the strengthening lemma. 
   uses the second property, Strengthening, of 
   Proposition~\ref{prop.metatheoryPTS} in an essential way. 
%%%However, previous research has learned that this 
%%%lemma is non-trivial to prove \cite{mckinna-pollack:jar1999}.
   This presents two difficulties: a practical one, since the existing
   formalisations of this lemma are highly
   non-trivial~\cite{mckinna-pollack:jar1999}; and a theoretical one,
   namely that we may \emph{not} be able to establish a correspondence
   between traditional and \gammainf{} presentations of a given type
   theory without first establishing strengthening.

More interestingly,  
   from the point of view of the pragmatics of formalisation,  
for this direction it is essential to abstract over
the kinds of free variables used in the definition of PTS judgments. At first this
does not seem troublesome, however, many definitions and theorems do not just depend on
the kind of free variables but also on finite sets of free variables. Hence we are 
also required to abstract over various operations on such finite sets. 
The recently developed finite set library 
%%%based on the new type classes feature in Coq \cite{lescuyer:containers} 
   \cite{lescuyer:containers}, based on the new type classes feature in Coq, 
might be very useful in implementing this abstraction.

\paragraph{Acknowledgments} 
Thanks to Jean-Christophe Filli\^atre for details about the
architecture of the Coq kernel.
We are grateful to the anonymous referees, and to the audience at LFMTP, especially John Boyland, Brigitte Pientka and Andrew Pitts, who made several helpful remarks, especially concerning related work. 
%%%
This research is partially funded by 
the NWO BRICKS/FOCUS project \realquote{ARPA: Advancing the Real use of Proof Assistants} 
and the NWO cluster \textsc{Diamant}. 

\bibliographystyle{eptcs}
\bibliography{Gi}

\end{document}